\documentclass[letter,11pt]{article}
\pdfoutput=1

\usepackage{jheppub} 

\usepackage[T1]{fontenc} 
\usepackage{amsthm}
\usepackage{slashed}
\usepackage{epsfig}
\usepackage{comment}
\usepackage{cancel}
\usepackage{bbm}
\usepackage{array}
\usepackage{bigints}
\usepackage{booktabs}
\usepackage{color}
\usepackage{dsfont}
\usepackage{float}
\usepackage{framed}
\usepackage{graphicx}
\usepackage{indentfirst}
\usepackage[mathscr]{eucal}
\usepackage{multirow}
\usepackage{tikz}
\usepackage{tikz-feynman,contour}
\usetikzlibrary {positioning}
\definecolor {processblue}{cmyk}{0.96,0,0,0}
\usetikzlibrary{arrows.meta}

\usepackage{subdepth}
\usepackage{titlesec}
\usepackage[dotinlabels]{titletoc}
\usepackage{wrapfig}
\usepackage[all]{xy}
\usepackage[vcentermath]{youngtab}
\usepackage{relsize}
\usepackage{hyperref}
\usepackage{makecell}  
\usepackage[table]{colortbl}

\numberwithin{equation}{section}

\usepackage{subfig}
\usepackage[utf8]{inputenc}
\usepackage{slashed}
\usepackage{amsmath}
\usepackage{amsfonts}
\usepackage{amssymb}
\usepackage{mathtools}

\usepackage{cleveref}
\crefname{section}{§\!\!}{§§\!\!}
\Crefname{section}{§}{§§}
\crefname{appendix}{Appendix\!}{Appendices\!}
\crefname{figure}{Fig.\,\!\!}{Figs.\,\!\!}

\newtheorem{lemma}{Lemma}
\newtheorem{theorem}{Theorem}

\newtheorem{defn}{Definition}

\newtheorem{coro}{Corollary}

\newtheorem{rmk}{Remark}

\def\0{{(0)}}
\def\1{{(1)}}
\def\2{{(2)}}

\def\<{\langle }
\def\>{\rangle }

\newcommand{\bea}{\begin{eqnarray}}
\newcommand{\eea}{\end{eqnarray}}
\newcommand{\be}{\begin{equation}}
\newcommand{\ee}{\end{equation}}
\newcommand{\ba}{\begin{align}}
\newcommand{\ea}{\end{align}}

   \makeatletter
  \let\over=\@@over \let\overwithdelims=\@@overwithdelims
  \let\atop=\@@atop \let\atopwithdelims=\@@atopwithdelims
  \let\above=\@@above \let\abovewithdelims=\@@abovewithdelims
\renewcommand\section{\@startsection {section}{1}{\z@}%
                                   {-3.5ex \@plus -1ex \@minus -.2ex}
                                   {2.3ex \@plus.2ex}%
                                   {\normalfont\large\bfseries}}

\renewcommand\subsection{\@startsection{subsection}{2}{\z@}%
                                     {-3.25ex\@plus -1ex \@minus -.2ex}%
                                     {1.5ex \@plus .2ex}%
                                     {\normalfont\bfseries}}

\newcommand{\beq}{\begin{equation}}
\newcommand{\eeq}{\end{equation}}
\newcommand{\beqa}{\begin{eqnarray}}
\newcommand{\eeqa}{\end{eqnarray}}
\newcommand{\beqar}{\begin{eqnarray*}}

\def\[{\big[}
\def\]{\big]}

\def\min{\text{min}}
\def\a{\alpha}
\def\b{\beta}
\def\G{\Gamma}

\def\a{{\alpha}}
\def\b{{\beta}}

\def\t{{\theta}}

\def\be{{\bar \epsilon}}

\def\CA{{\mathcal A}}

\def\CC{{\mathcal C}}

\def\CH{{\mathcal H}}

\def\SB{{\mathscr B}}

\def\+{{(+)}}
\def\-{{(-)}}
\def\0{{(0)}}
\def\1{{(1)}}
\def\2{{(2)}}
\def\3{{(3)}}
\def\4{{(4)}}
\def\5{{(5)}}

\def\S{{\tt{S}}}

\def\k{{\tt{K}}}
\def\l{\lambda}

\def\N{{\sf N}}
\def\D{{\sf D}}
\def\R{{\sf R}}

\def\polyindfont#1{\mathscr{#1}}
\def\pI{\polyindfont{I}}
\def\pJ{\polyindfont{J}}
\def\pK{\polyindfont{K}}

\def\ii{{\underline{i}}}
\def\jj{{\underline{j}}}

\def\pII{{\underline{\pI}}}

\def\gI{\Gamma}

\def\SB#1{{\hat e}^{\, #1}}
\def\IB#1{{\hat f}^{\, #1}}
\def\KB#1{{\hat g}^{#1}}

\def\av{\omega}

\def\IQfont#1{{\mathbf #1}}
\def\Q{{\IQfont{Q}}}

\def\I{{\tt I}}

\definecolor{vecolor}{rgb}{0.7,0.3,0.9}

\definecolor{rust}{rgb}{0.8,0.2,0.2}

\def\shadeI{\cellcolor{blue!5}}
\def\shadeK{\cellcolor{red!5}}

\title{\boldmath 
Superbalance of Holographic Entropy Inequalities}
 
\author[]{Temple He}
\author[]{, Veronika E.\ Hubeny}
\author[]{, Mukund Rangamani}

\affiliation[]{Center for Quantum Mathematics and Physics (QMAP)\\
Department of Physics, University of California, Davis, CA 95616 USA}

\emailAdd{tmhe@ucdavis.edu}
\emailAdd{veronika@physics.ucdavis.edu}
\emailAdd{mukund@physics.ucdavis.edu}

\abstract{ The domain of allowed von Neumann entropies of a holographic field theory carves out a polyhedral cone -- the holographic entropy cone -- in entropy space. Such polyhedral cones are characterized by their extreme rays.  For an arbitrary number of parties, it is known that the so-called perfect tensors are extreme rays. In this work, we constrain the form of the remaining extreme rays by showing that they correspond to geometries with vanishing mutual information between any two parties, ensuring the absence of Bell pair type entanglement between them. This is tantamount to proving that besides subadditivity, all non-redundant holographic entropy inequalities are superbalanced, i.e.\ not only do UV divergences cancel in the inequality itself (assuming smooth entangling surfaces), but also in the purification thereof.
}

\begin{document} 
\maketitle
\flushbottom

\section{Introduction}
\label{sec:intro}

Entropy inequalities in quantum information theory are linear constraints on the von Neumann (entanglement)  entropies of the various subsystems
that give a useful characterization of the full system. For instance, given a factorizable Hilbert space $\CH_1 \otimes \CH_2$, the entropy of the combined system is constrained by those of the individual subsystems $\CH_1$ and $\CH_2$ via the subadditivity (SA) inequality
\begin{align}\label{sa}
	\S_1 + \S_2 \geq \S_{12}\ ,
\end{align}	
where $\S_i$ is the entanglement entropy associated to $\CH_i$, and $\S_{ij}$ is that associated to $\CH_i \otimes\CH_j$. Equivalently, \eqref{sa} can be expressed as the non-negativity of the mutual information
\begin{align}
	\I_{12} \equiv \S_1 + \S_2 - \S_{12}\ ,
\end{align}
which characterizes the total amount of correlation between $\CH_1$ and $\CH_2$. 
Understanding the general structure of such entropy inequalities provides intuition about the nature of admissible quantum states and thus remains of broad interest both in quantum mechanics and in continuum quantum field theories (QFTs).

Progress in this direction has however proven difficult, since for general quantum systems it remains unclear, even for those involving a small number of parties, what the complete set of the entropy inequalities are, or even whether there are finitely many. Moreover, in QFTs the computation of entanglement entropy is notoriously difficult.  In recent years, progress has been made by considering a subclass of  QFTs that are holographic conformal field theories (CFTs), which are dual to a gravitational theory in an asymptotically AdS spacetime.
For this class of theories, the computation of entanglement entropy is geometric; it is given by the area of an extremal (HRT) surface in the bulk  geometry  \cite{Ryu:2006bv,Hubeny:2007xt}. 
The geometric framework allows one to analyze entropic constraints and obtain a class of  holographic entropy inequalities (HEIs). Such inequalities are not universally true for an arbitrary quantum state, but they do  hold for any  `geometric state' in a holographic CFT (i.e.\ one with a geometric bulk dual).  This means that such relations can in turn be utilized to characterize geometric states. The simplest HEI is the monogamy of mutual information (MMI) and is given by \cite{Hayden:2011ag}
\begin{align}\label{mmi}
	\S_{12} + \S_{13} + \S_{23} \geq \S_1 + \S_2 + \S_3 + \S_{123}\ .
\end{align}
The combination of subsystem entropies appearing here is the same as that in topological entanglement entropy \cite{Kitaev:2005dm}, and is tantamount to the non-positivity of the tripartite information $\I_{123}$ (explicitly defined in \eqref{e:In}).

A systematic study of HEIs for an arbitrary number of parties was initiated by \cite{Bao:2015bfa}, and continued via a different approach in \cite{Hubeny:2018trv,Hubeny:2018ijt}. Specifically, consider a set of disjoint spatial regions in the boundary CFT, labeled by $\CA_1,\ldots,\CA_\N$. Given any holographic CFT state, the entanglement entropies involving the various collections of spatial regions can be collectively associated to a vector in entropy space, where the components are simply the entropies of all the composite subsystems.\footnote{
The actual values are infinite (they suffer from the area-law divergence) whenever the system of interest is bounded by an entangling surface on the boundary (as opposed to covering an entire spatial slice of the CFT).  Therefore, when necessary we introduce a UV cutoff, which renders the regulated entropies (as well as their ratios) cutoff-dependent. This issue was circumvented in \cite{Hubeny:2018trv,Hubeny:2018ijt} by working with  `proto-entropies' (corresponding to the HRT surfaces themselves rather than their areas) and analyzing their formal relations, and we leave the recasting of our results in terms of proto-entropies to future work.
} For instance, in the case $\N=3$, the entropy vector is
\begin{align}
\label{eq:Svec}
	\vec\S = (\S_1,\ \S_2,\ \S_3,\ \S_{12},\ \S_{13},\ \S_{23},\ \S_{123})\ ,
\end{align}
where $\S_i \equiv \S_{\CA_i}$, $\S_{ij} = \S_{\CA_i\cup \CA_j}$, and so forth. In general, given $\N$ parties, the entropy space has $\D=2^{\N}-1$ dimensions, and constraints such as SA or MMI carve out allowable half-spaces of the entropy space in which the entropy vectors can reside.  The authors of \cite{Bao:2015bfa} argued that the physically allowed region forms a convex cone in the entropy space, dubbed the holographic entropy cone (HEC). This was accomplished by restricting attention to the static context and using the minimal surface (RT) prescription of \cite{Ryu:2006bv} to recast the problem using graph theoretic language, so the problem of determining the HEC is in fact equivalent to finding HEIs. In the more general framework of \cite{Hubeny:2018ijt}, the intersection of all half-spaces determined by the {\it primitive}\footnote{
Primitive HEIs, defined in \cite{Hubeny:2018trv}, are ones that capture an independent type of correlation, which automatically eliminates the redundant HEIs such as strong subadditivity.  
} HEIs involving $\N$ parties was for clarity dubbed the {\it holographic entropy polyhedron} to emphasize that this construct is defined by the linear HEIs using proto-entropies as opposed to the allowed rays in the entropy space. Nevertheless, it is believed (and known in the static case for $\N \le 5$) that the polyhedron indeed coincides with the HEC.\footnote{ 
In certain restricted settings it has been shown that the static (or RT) HEC and the more general HRT cone coincide as well; see \cite{Bao:2018wwd,Czech:2019lps}. However, a general proof has thus far proven elusive.}
Currently, the HEC (at least for static configurations) is fully known up to $\N=5$ \cite{Bao:2015bfa, Cuenca:2019uzx}, with ongoing attempts to determine the entropy cone for larger $\N$.

In this work, we take a step towards the characterization of the static HEC for arbitrary $\N$. We argue that apart from SA, all non-redundant HEIs, or HEIs that are each not a sum of two other HEIs, must obey a property known as \emph{superbalance} \cite{Hubeny:2018ijt}. Roughly speaking, this notion requires that the UV divergences in the entanglement entropies of 
disjoint subsystems with smooth entangling surfaces cancel out not only in the combination of entropies entering the inequality, but also when we consider purifications of the inequality. The idea behind the purification, further explained in  \cref{s:superbalance}, is to replace one of the subsystems with the purifier, i.e.\ the complement of the union of all the $\N$ parties.  Although purification is a symmetry of the full system, it is not manifest when we represent HEIs in terms of the subsystem entanglement entropies.  For this reason it is more practical to work in a representation of the entropy space basis consisting of perfect tensors (PTs), the so-called K-basis introduced in \cite{He:2019ttu}, which has many additional useful properties.  We will therefore adopt this representation as a technical aid in the proof.  Along the way, we will also utilize the I-basis presented in \cite{Hubeny:2018ijt,He:2019ttu} as a convenient intermediate step.
The main result, that all genuinely multipartite HEIs (i.e.\ excluding SA) are superbalanced, vastly simplifies the explicit search for further HEIs and validates the generality of the ``sieve'' method of \cite{Hubeny:2018ijt}.  

The outline of the paper is as follows. We give a quick overview of the  entropy space and bases thereof in  \cref{s:K-basis}.  In  \cref{s:proof} we review the notions of balance and superbalance and present the main  technical argument. We conclude with a short discussion in  \cref{s:discussion}.

\section{Three representations of entropy space}
 \label{s:K-basis}

The entropy space is a vector space, and can be represented by a variety of bases  that span the space. We quickly review three particularly useful bases explored in \cite{He:2019ttu}, dubbed the S, I, and K bases. These bases each have their advantages and disadvantages  (we refer the reader to \cite{He:2019ttu} for further details); in what follows, we will make particular use of the K-basis representation.

Consider the entropy space associated to $\N$ parties labeled $\CA_1,\ldots, \CA_{\N}$, and let $\pI \subseteq [\N] \equiv \{1,\ldots,\N\}$ be the index denoting the nonempty subsets of the $\N$ parties. Following the terminology introduced in \cite{Hubeny:2018trv}, we will refer to the subsystems involving a single party ($|\pI| = 1$) as {\it monochromatic}, and to general composite subsystems involving possibly multiple parties ($|\pI| \geq 1$) as  {\it polychromatic}.

When characterizing an entropy vector, it is natural to use the orthonormal basis $\big\{ \SB{\pI} \big\}$ such that
\begin{align}
	\vec \S = \sum_{\pI\subseteq [\N]} \S_\pI \,  \SB{\pI}\ , \quad\text{where $\S_\pI \in \mathbb R$\ .}
\end{align}
This means $\SB{\pI}$ is simply the vector given by $1$ in the component associated to the polychromatic subsystem $\pI$ and $0$ elsewhere, and it is the most common basis used when studying entropy space. This was dubbed the S-basis in \cite{He:2019ttu}. 

A second (and in many ways more convenient) basis used when studying entropy cones is the I-basis. For any $\pI \subseteq [\N]$, multipartite information involving the monochromatic parties in $\pI$ is defined as \cite{Hubeny:2018ijt}
\begin{align}
\label{e:Ingen}
	\I_\pI = \sum_{\pK \subseteq \pI} (-1)^{1+|\pK|} \, \S_\pK\ .
\end{align}
For instance, the expressions for $n$-partite information for $n=1,2,3$ are 
\begin{align}
\begin{split}
	\I_i &= \S_i \\
	\I_{ij} &= \S_i + \S_j - \S_{ij} \\
	\I_{ijk} &= \S_i + \S_j + \S_k - \S_{ij} - \S_{ik} - \S_{jk} + \S_{ijk}\ .
\end{split}
\label{e:In}
\end{align}
Thus, 1-partite information is just the entanglement entropy, 2-partite information the mutual information, 3-partite information the tripartite information, and so on. In the I-basis $\big\{ \IB{\pI}\big\}$, we choose the orthonormal basis vectors $\IB{\pI}$ such that
\begin{align}
	\vec \S = \sum_{\pI \subseteq [\N]} \I_\pI \,\IB{\pI}\ , \quad\text{where $\I_\pI \in \mathbb R$\ .}
\end{align}

Finally, we introduce the K-basis, which was first explored in \cite{He:2019ttu} and will be used extensively in our proofs. The K-basis is comprised of the orthonormal entropy vectors $\big\{ \KB{\G} \big\}$ corresponding to perfect tensors (PTs). A PT$_{2s}$ is a pure state involving $2s$ parties, with $s \in \mathbb Z^+$, such that the reduced density matrix associated to any $s$ parties is maximally mixed. We can capture its entropy structure using graphs such as those depicted in \cref{f:Kbasis}. The Bell pair (BP) $\KB{\{ij\}}$, for instance, is simply a PT$_2$ and is represented by a single link joining two vertices corresponding to the monochromatic systems $i$ and $j$.

\begin{figure}[h]
\centering
    \begin{tikzpicture}[
       thick, scale=0.75,
       acteur/.style={
         circle,
         fill=black,
         thick,
         inner sep=2pt,
         minimum size=0.2cm
       }
     ]
	\node (a1) at (0,2) [acteur,label=$\CA_1$]{}; 
      \node (a2) at (1,1) [acteur,label=right:$\CA_2$]{};
      \node (a3) at (0,0) [acteur,label=below:$\CA_3$]{}; 
      \node (a4) at (-1,1) [acteur,label=left:$\CA_4$]{}; 
      
      \draw (a1) -- (a3);
\end{tikzpicture} \qquad\qquad     \begin{tikzpicture}[
       thick,scale=0.75,
       acteur/.style={
         circle,
         fill=black,
         thick,
         inner sep=2pt,
         minimum size=0.2cm
       }
     ] 
	\node (a1) at (0,2) [acteur,label=$\CA_1$]{}; 
      \node (a2) at (1,1) [acteur,label=right:$\CA_2$]{};
      \node (a3) at (0,0) [acteur,label=below:$\CA_3$]{}; 
      \node (a4) at (-1,1) [acteur,label=left:$\CA_4$]{}; 
      \node (b) at (0,1) [acteur,label={}] {};
      
       \draw (a1) -- (b);
       \draw (a2) -- (b);
       \draw (a3) -- (b);
      	\draw (a4) -- (b);
\end{tikzpicture} 
\caption{These discrete graphs capture the entanglement structure of a PT$_2$ (the BP $\KB{\{13\}}$) on the left, and a PT$_4$ ($\KB{\{1234\}}$) on the right. The boundary vertices denote the individual parties, and each edge represents a single unit of entanglement. The holographic entanglement entropy associated to a set of boundary vertices is the minimum cut associated to those vertices.
}
\label{f:Kbasis}
\end{figure}
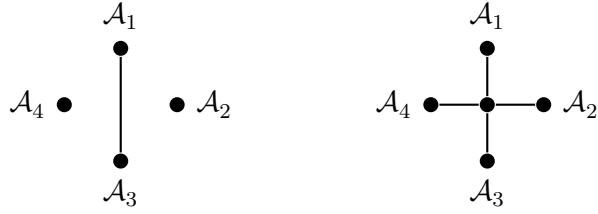

Because these PTs are pure states, in order to use the entanglement entropy associated to PTs as basis vectors, we need to include the purification of our original state, which we denote $\CA_{\N+1} = \overline{\bigcup_{i=1}^\N \CA_i}$. The K-basis is then given by the entropy vectors corresponding to all the even-party PTs involving $[\N+1]$ parties.\footnote{
As a sanity check, we observe that there are exactly $\D = 2^\N-1$ such PTs, which follows from the identity
\begin{align}
	\sum_{s=1}^{\lceil \N/2 \rceil} \binom{\N+1}{2s} = 2^\N - 1\ .
\end{align}} 
As an example, for any $\N \geq 3$, the two graphs in \cref{f:Kbasis} are basis vectors in the K-basis, and can be written in the S-basis as
\begin{align}\label{gvecs}
\begin{split}
	\text{BP (PT$_2$)}\ \ \KB{\{13\}} : \quad & (\S_1,\S_2,\S_3,\S_{12},\S_{13},\S_{23},\S_{123}) = (1,0,1,1,0,1,0) \\
	\text{PT$_4$}\ \ \KB{\{1234\}} : \quad & (\S_1,\S_2,\S_3,\S_{12},\S_{13},\S_{23},\S_{123}) = (1,1,1,2,2,2,1)\ .
\end{split}
\end{align}
We will use the index $\G$ to specify a given K-basis vector; in particular, 
 $\G \subseteq [\N+1]$ with $|\G|$ even. For example,
the expressions appearing in \eqref{gvecs} correspond to  $\G = \{13\}$ and $\G = \{1234\}$, respectively.
For future reference, we will designate the monochromatic parties in $[\N]$ with lower case Latin indices $i,j,\ldots$, and polychromatic parties in $[\N]$ with script Latin indices $\pI, \pK,\ldots$. If we wish to include the purifier, we will underline the index, e.g.  $\ii \in [\N+1]$. 

Writing the entropy vector for an $\N$-party system in terms of the K-basis, we have\footnote{
Since the structural form of the HEIs in the K-basis depends on $\N$, we included the superscript $(\N)$ in the coefficients of $\KB{\G}$ explicitly to avoid confusion, following the notation of \cite{He:2019ttu}.
}
\begin{align}
	\vec \S = \sum_{\G \subseteq [\N+1]} \k_\G^{(\N)} \, \KB{\G}\ , \quad\text{where $\k_\G^{(\N)} \in \mathbb R$} \ .
\end{align}
The K-basis is particularly convenient since as was proved in \cite{He:2019ttu}, the even-party PTs $\KB{\G}$ are extreme rays of the HEC, meaning that they cannot be expressed as a positive combination of any other vectors within the cone.  This has a useful consequence for the form of HEIs rendered in the K-basis.

Consider an {\it information quantity}, written in any of our three bases of interest: 
\begin{equation}
\label{eq:SIKineq}
\Q
=\sum_{\pI \subseteq [\N]}\mu_{\pI}\,\S_{\pI}
=\sum_{\pI \subseteq [\N]}\nu_{\pI}\,\I_{\pI}
=\sum_{\gI \subseteq [\N+1]}\lambda_{\gI}\, \k_{\gI}^{(\N)} \ ,
\end{equation}	
so $\Q$ is equivalently specified by rational coefficients $\mu_{\pI}$ in the S-basis, $\nu_{\pI}$ in the I-basis, or $\lambda_{\gI}$ in the K-basis.
An HEI simply identifies a sign-definite information quantity, expressed as  $\Q\ge 0$.
Since $\KB{\G}$ are extreme rays of the HEC and hence lie within the cone, it was shown in \cite{He:2019ttu} that $\l_\G \geq 0$ in all the HEIs. To exemplify this, we have included in Table~\ref{t:N=5} (taken from \cite{He:2019ttu}) a list of the $\N=5$ non-redundant entropy inequalities (or rather a single representative from each  permutation and purification symmetry orbit) in the three bases. 
\begin{table}
\begin{center}
\scriptsize
\begin{tabular}{| c | c |l  | }
\hline
Relation & \makecell{Basis}  &  Holographic Entropy Inequality  \\ 
\hline
\hline
 SA$_{(1,1)}$
 &  S  &  $\S_{1}+\S_{2}-\S_{12}$ \\ 
 &  \shadeI{I } & \shadeI{$\I_{12}$}  \\
 &  \shadeK{K} & \shadeK{$\k_{12}^{\5}$ } 
\\  \hline  
 MMI$_{(1,1,1)}$
 &  S  & $-\S_{1}-\S_{2}-\S_{3}+\S_{12}+\S_{13}+\S_{23}-\S_{123}$   \\ 
&  \shadeI{I} & \shadeI{$-\I_{123}$}  \\
 &  \shadeK{K} & \shadeK{$\k_{1234}^{\5}+\k_{1235}^{\5}+\k_{1236}^{\5}$ } 
\\  \hline 
 MMI$_{(1,2,2)}$
 &  S  & $-\S_{1}-\S_{23}-\S_{45}+\S_{123}+\S_{145}+\S_{2345}-\S_{12345}$  \\   
&  \shadeI{I } & \shadeI{$-\I_{124}-\I_{125}-\I_{134}-\I_{135}+ \I_{1234} + \I_{1235} + \I_{1245}+ \I_{1345}- \I_{12345}$}  \\
 &  \shadeK{K } & \shadeK{$\k_{1246}^{(5)}+\k_{1256}^{(5)}+\k_{1346}^{(5)}+\k_{1356}^{(5)}+\k_{123456}^{\5}
$ } 
\\ \hline 
$Q^{(5)}_1$
 &  S   &
   $ - \S_{12} - \S_{23} - \S_{34}  - \S_{45} - \S_{15} + \S_{123} + \S_{234} + \S_{345} + \S_{145} + \S_{125}  - \S_{12345}$ \\ 
 &  \shadeI{I  } & \shadeI{$-\I_{124}   - \I_{134} - \I_{135} - \I_{235} - \I_{245} + \I_{1234} + \I_{1235} + \I_{1245} + \I_{1345} + \I_{2345} - \I_{12345}$}  \\
 &  \shadeK{K } & \shadeK{$\k_{1246}^{(5)} + \k_{1346}^{(5)} + \k_{1356}^{(5)} + \k_{2356}^{(5)} + \k_{2456}^{(5)} +2 \k_{123456}^{(5)}  $ } 
\\  \hline 
$Q^{(5)}_2$
 &  S   &  \makecell{$- \S_{12}- \S_{13} - \S_{14} - \S_{23} - \S_{25} - \S_{45}+2\S_{123} + \S_{124} + \S_{125} + \S_{134} + \S_{145} + \S_{235} + \S_{245}   $\\$- \S_{1234} - \S_{1235} - \S_{1245}$}  \\ 
  &  \shadeI{I  } & \shadeI{$- \I_{124} - \I_{125} - \I_{135} - \I_{234} + \I_{1234} + \I_{1235} + \I_{1245} $}  \\
 &  \shadeK{K } & \shadeK{$ \k^{\5}_{1246} + \k^{(5)}_{1256} + \k^{\5}_{1345} + \k^{\5}_{1356} + \k^{\5}_{2345} + \k^{\5}_{2346} + 3 \k^{\5}_{123456} $ } 
\\  \hline 
$Q^{(5)}_3$
 &  S  &   \makecell{$ - \S_{12} - \S_{13} - \S_{14} - \S_{25}- \S_{35} - \S_{45} + \S_{123}  + \S_{124} + \S_{125} + \S_{134} + \S_{135} + \S_{145} + \S_{235} $\\$  + \S_{245} + \S_{345}  - \S_{234} - \S_{1235} - \S_{1245} - \S_{1345}$}  \\ 
 &  \shadeI{I} & \shadeI{$- \I_{125} - \I_{135} - \I_{145} - \I_{234} + \I_{1235} + \I_{1245} + \I_{1345} $}  \\
 &  \shadeK{K } & \shadeK{$ \k^{\5}_{1234} + \k^{\5}_{1256} + \k^{\5}_{1356} + \k^{(5)}_{1456} + \k^{\5}_{2345} + \k^{\5}_{2346} +3 \k^{\5}_{123456} $ } 
\\  \hline 
$Q^{(5)}_4 $
 &  S    &   \makecell{$- \S_2 - \S_3 - \S_4 - \S_5 - \S_{12} - \S_{13} + \S_{23} + \S_{45} + \S_{123} + \S_{124} + \S_{125} + \S_{134} + \S_{135} - \S_{145}$\\$ - \S_{1234} - \S_{1235}$}  \\ 
 &  \shadeI{I} & \shadeI{$- \I_{123} - \I_{145} - \I_{234} - \I_{235} + \I_{1234} + \I_{1235} $}  \\
 &  \shadeK{K } & \shadeK{$ \k^{\5}_{1236} + \k^{\5}_{1245} + \k^{\5}_{1345} + \k^{\5}_{1456}  + 2 \k^{\5}_{2345}  + \k^{\5}_{2346} + \k^{\5}_{2356} + 2 \k^{\5}_{123456} $ } 
\\  \hline 
$Q^{(5)}_5 $
 &  S   &  \makecell{$ - 2\S_{12} - 2\S_{13} - \S_{14} - \S_{15} - \S_{23}  - 2\S_{24} - 2\S_{35}  - \S_{45}+ 3\S_{123} + 3\S_{124} + \S_{125} + \S_{134} $\\$  + 3\S_{135}+ \S_{145} + \S_{234} + \S_{235} + \S_{245} + \S_{345} - 2\S_{1234} - 2\S_{1235} - \S_{1245} - \S_{1345}$}   \\ 
  &  \shadeI{I  } & \shadeI{$- \I_{123} -2 \I_{125} - 2 \I_{134} - \I_{145} - \I_{234} - \I_{235} + 2\I_{1234} +2 \I_{1235} + \I_{1245} + \I_{1345} $}  \\
 &  \shadeK{K } & \shadeK{$\k^{\5}_{1236} + \k^{(5)}_{1245} + 2 \k^{\5}_{1256} + \k^{\5}_{1345} + 2 \k^{\5}_{1346} + \k^{\5}_{1456}  + 2 \k^{\5}_{2345}    + \k^{\5}_{2346} + \k^{\5}_{2356} + 6 \k^{\5}_{123456}$ } 
\\  \hline 
\end{tabular}
\end{center}
\caption{The non-redundant HEIs for $\N=5$ in the S, I, and K bases. The first three rows are the standard SA and MMI (with the subscript corresponding to the cardinality of the arguments as in \cite{He:2019ttu}), and the last five rows are inequalities discovered in \cite{Bao:2015bfa}, and labeled simply as $Q^{(5)}_i $ since their deeper meaning is as yet obscure. Note that the coefficients of $\k^{(5)}_\G$ in all the HEIs are nonnegative, exemplifying a more general result of \cite{He:2019ttu}.}
\label{t:N=5}
\end{table}

A closer examination of Table~\ref{t:N=5} reveals that except for SA, all the other non-redundant HEIs do not involve $\k_{\ii\jj}^{(5)}$ for all $\ii,\jj \in [\N+1]$. Therefore, one may suspect that for any $\N$, all non-redundant $\N$-party HEIs other than SA will not involve $\k_{\ii\jj}^{(\N)}$. 
We will prove that this is indeed true, and as we shall see, this is tantamount to the statement that non-redundant HEIs other than SA obey a property known as superbalance.

\section{Multipartite holographic entropy inequalities are superbalanced} 
\label{s:proof}

Having  described the various bases characterizing the entropy space, we turn to the main point of the discussion -- demonstrating that all non-redundant HEIs except SA are superbalanced. We first formally introduce the notion of superbalance from \cref{s:superbalance}, which was originally expressed in terms of the I-basis. In \cref{ss:Iofg}, we reformulate the notion of superbalance in terms of the K-basis instead (equivalently, finding the basis transform between the I-basis and K-basis), since it will be a helpful intermediate step in proving the main result. Armed with these results, we state and prove our main theorem in  \cref{ss:proof}. A few additional details regarding the map between I and K bases can be found in  \cref{a:IKconf}.

\subsection{Balance and superbalance}
\label{s:superbalance}

We begin by reviewing the notion of an HEI being balanced and superbalanced \cite{Hubeny:2018ijt}. Consider an entropy inequality in S-basis of the form $\Q\ge0$, with $\Q$ given by the first expression in \eqref{eq:SIKineq}.  For every $i \in [\N]$, define the function
\begin{align}
	\t_i(\pI) = \begin{cases}
	1 & i \in \pI \\
	0 & i \notin \pI\ .
\end{cases}
\end{align}
Then an HEI is {\it balanced} if the occurrence of each party individually cancels, i.e.\ if it satisfies 
\begin{align}\label{balance}
	\sum_{\pI \subseteq [\N]} \mu_\pI \, \theta_i(\pI) = 0 \qquad \forall \ i \in [\N] \ .
\end{align}
For example, as the reader can easily verify, all the HEIs listed in Table~\ref{t:N=5} are balanced.

However, even if a given HEI is balanced, its purification need not necessarily be balanced. For instance, if we use purification on SA$_{(1,1)}$ from the first row of Table~\ref{t:N=5} to eliminate party $1$, i.e. we replace $\S_1$ with $\S_{23456}$ and $\S_{12}$ with $\S_{3456}$ since in a pure state $\S_\pII = \S_{[\N+1]\setminus\pII}$, the resulting inequality is the Araki-Lieb (AL) inequality
\begin{align}
	\S_{23456} + \S_{2} - \S_{3456} \geq 0\ ,
\end{align}
which is manifestly not balanced.  Since the HEC is symmetric under all permutations of the subsystem labels, including the purifier, it is more useful to consider the class of HEIs that are not only themselves balanced, but whose purifications are likewise also balanced. Such inequalities are defined to be {\it superbalanced} \cite{Hubeny:2018ijt}.\footnote{
In other words, HEIs that remain balanced under any purification are defined to be superbalanced.
} 
Superbalanced HEIs have  structurally useful properties, which are further explained in \cite{Hubeny:2018ijt}.  In particular, the methods used to generate new HEIs are best tailored to this case.  

In \cite{Hubeny:2018ijt}  the notion of balance and superbalance was generalized to the notion of $\R$-balance, where $\R$ is a positive integer.\footnote{
The property of being balanced would correspond to $\R= 1$ and that of being superbalanced to $\R = 2$.
}
In particular, an information quantity $\Q$ is {\it $\R$-balanced} if its expansion in the I-basis does not contain terms of rank $\R$ or lower.  In the I-basis (middle expression of \eqref{eq:SIKineq}), superbalanced HEIs would then contain $\I_\pI$ with $|\pI|\ge 3$, as is manifested in  Table~\ref{t:N=5}.

Hitherto, it has been an open question as to whether superbalance is a true property of all non-redundant HEIs other than SA (which implicitly includes AL in its symmetry orbit).  Table~\ref{t:N=5}
manifests the result for $\N=5$, and it is easy to verify that it continues to hold for all the HEIs discovered so far. In  \cref{ss:proof}, we shall prove the result in full generality. Before doing so, it is first useful to phrase the notion of superbalance directly in the K-basis, which we now turn to in \cref{ss:Iofg}.

\subsection{Multipartite information of PTs}
\label{ss:Iofg}

Since the K-basis is particularly convenient for discussing HEIs as evidenced above, it will be useful to re-express the notion of $\R$-balance in terms of the K-basis. Of course, one could determine the explicit conversion between the I and K bases by evaluating the multipartite information of the PTs, as is explicitly carried out in \cref{a:IKconf}.  However, here it suffices to obtain a simple criterion for a given matrix element to vanish. Because the K-basis invokes the purifier, we will for aesthetic convenience label it by $\ii=0$ (as opposed to $\N+1$) in the rest of this subsection, and similarly let $[\N+1] \equiv \{0,1,\ldots,\N\}$ (as opposed to $\{1,\ldots,\N+1\}$).

\begin{lemma}\label{l:Iofg}
The multipartite information $\I_{\pI}$ of a PT $\KB{\G}$ vanishes whenever $\pI$ evokes any subsystem which is not contained in $\G$, i.e.\
\begin{equation} 
\label{eq:Iofgzero}
\left.\I_\pI \right|_{\KB{\G}} = 0 
\qquad {\rm if} \ \
 \pI \nsubseteq \G \ .
\end{equation}	
\end{lemma}
\noindent
Note that since $\pI$ never contains the purifier $0$, we can equivalently write the condition in \eqref{eq:Iofgzero} as $\pI \nsubseteq \gI\setminus 0$, where $\gI\setminus 0$ indicates we exclude the purifier from $\gI$ when present.
\begin{proof}
Suppose $\pI \nsubseteq \gI$.  
This means there must be at least one party, call it $i$, that is contained in $\pI$ but not in $\G$.  Let us (with a slight abuse of notation) denote  the rest of $\pI$ by $\pJ \equiv \pI \setminus i$ and write $\pI = i \pJ$, where the union is implied as usual.
Since $ i \notin \G$, the entropy vector for $\KB{\G}$ has the property that $\left.\S_{i\pK} \right|_{\KB{\G}} \, = \left.\S_{\pK} \right|_{\KB{\G}} $ for any $\pK \subseteq [\N+1]$ and $\left.\S_i \right|_{\KB{\G}}  =0$.
Splitting \eqref{e:Ingen} into terms not containing $i$ and terms containing $i$, we have
\begin{align}
\begin{split}
\left.\I_\pI \right|_{\KB{\G}}&= 
	\left.\left( \sum_{\pK \subseteq \pJ}
	 (-1)^{1+|\pK|} \, \S_\pK
	+ \sum_{\pK \subseteq \pJ} (-1)^{2+|\pK|} \, \S_{i\pK}
	+ \S_i \right)\right|_{\KB{\G}}
	\\
	& =  \sum_{\pK \subseteq \pJ} (-1)^{1+|\pK|} \, \underbrace{\left(\left.\S_\pK\right|_{\KB{\G}}-\left.\S_{i\pK}\right|_{\KB{\G}}\right)}_{=0} + \underbrace{\left.\S_i\right|_{\KB{\G}}}_{=0} 
	= 0  \ ,
\end{split}
\end{align}
thereby proving the statement.
\end{proof}

This observation has the immediate consequence that when we express the multipartite information $\I_\pI $ as a linear combination of $\k_\G$ so that
\begin{equation} \label{eq:IfromK}
\I_\pI = \sum_{\G \subseteq [\N+1]} \alpha_{\pI}{}^{\gI } \, \k_\G \ , 
\end{equation}	
the coefficients $\a_\pI{}^\gI$ vanish whenever the cardinality of $\pI$ is greater than that of $\gI \setminus 0$, i.e.
\begin{equation} \label{eq:Iofgmincard}
\alpha_{\pI}^{\ \gI} =0 \quad \text{whenever}\quad |\G\setminus 0| < |\pI| \  .
\end{equation}	
This follows from evaluating both sides of \eqref{eq:IfromK} on a PT $\KB{\G'}$ and using $\left.\k_\G\right|_{\KB{\G'}} = \delta_{\G\G'}$ to express the coefficients directly as 
\begin{equation} 
\label{eq:alphadef}
\alpha_{\pI}{}^{\gI} = \left.\I_\pI \right|_{\KB{\G}} \ .
\end{equation}	
If $|\G\setminus 0| < |\pI|$, $\G\setminus 0$ cannot contain the entire subsystem $\pI$, so by Lemma~\ref{l:Iofg}, $\alpha_{\pI}{}^{\gI}$ must vanish.

We will now use \eqref{eq:Iofgmincard} to convert the criterion of $\R$-balance from its I-basis definition to its K-basis definition. Recall that an information quantity $\Q$ is $\R$-balanced iff its representation in the I-basis only involves terms $\I_\pI$ with $|\pI|\ge\R+1$. Thus, when we convert to it the K-basis using \eqref{eq:IfromK}, we have
\begin{align}
\begin{split}
	\Q =\sum_{\pI \subseteq [\N] , |\pI|>\R}\nu_{\pI}\,\I_{\pI} 
	&= \sum_{\G \subseteq [\N+1]} \left( \sum_{\pI \subseteq [\N], |\pI|>\R} \nu_{\pI}\, \alpha_{\pI}^{\ \gI}\right)  \, \k_\G^{(\N)}  \\
	&= \sum_{\G \subseteq [\N+1], |\G\setminus 0|\ge|\pI|} \left( \sum_{\pI \subseteq [\N], |\pI|>\R} \nu_{\pI}\, \alpha_{\pI}^{\ \gI}\right)  \, \k_\G^{(\N)} \ ,
\end{split}
\end{align}	
where in the last step we used \eqref{eq:Iofgmincard}. This means if $\Q$ is $\R$-balanced, it will only involve $\k_\gI^{(\N)}$ with $|\gI\setminus 0| \geq \R+1$.

Conversely, suppose $\Q$ takes the form 
\begin{equation}
\label{eq:KineqR}
\Q =\sum_{\gI \subseteq [\N+1], |\G\setminus 0| > \R}\lambda_{\gI}\, \k_{\gI}^{(\N)} 
\end{equation}	
for some arbitrary coefficients $\lambda_{\gI}$.  Comparing this with the last expression in \eqref{eq:SIKineq}, this is tantamount to requiring 
\begin{align}
	\lambda_{\gI} =0 \quad\text{whenever}\quad\begin{cases}
	|\G|\le\R & \text{if $0 \nsubseteq \G$} \\
	|\G|\le\R+1 & \text{if $0 \subset\G$} \ .
\end{cases}
\end{align}
To show that $\Q$ is in fact guaranteed to be $\R$-balanced, we need to re-express each $\k_{\G}^{(\N)}$ in \eqref{eq:KineqR} in the I-basis, formally written as
\begin{equation} \label{eq:KfromI}
 	\k_\G^{(\N)}  = \sum_{\pI \subseteq [\N]} \beta_{\G}{}^{\pI} \, \I_\pI \ ,
\end{equation}	
and argue that when we substitute this into \eqref{eq:KineqR}, the final expression does not involve $\I_\pI$ for $|\pI| \leq \R$. Note that the matrix ${\beta}$ involved in \eqref{eq:KfromI} is simply the inverse of the matrix  ${\alpha}$ from \eqref{eq:IfromK}.  

To relate the two matrices, it is more convenient to use an `inclusion-ordering' rather than the lexicographic ordering of \eqref{eq:Svec}.  For the I-basis, we order the $\pI$ indices according to 
\begin{equation}\label{eq:Iorder}
\{ 1,\, 2,\, 12,\, 3,\, 13,\, 23,\, 123,\, 4,\, 14,\, 24,\, 124,\, 34,\, 134,\, 234,\, 1234, \ldots \} \ .
\end{equation}	
In particular, we can construct the list iteratively: to arrange the indices of an $(\N+1)$-party entropy vector (excluding the purifier), we start by writing the indices $\{\pI\}$ of the $\N$-party entropy vector and append $i=\N+1$ followed by $\{ \pI i \}$.  The inclusion-ordering for $\G$ follows the same logic, but since $|\G|$ is restricted to be even, the expression looks slightly more complicated.  Nevertheless, we can use \eqref{eq:Iorder} with the simple rule that we leave $\pI$ as is if $|\pI|$ is even, but  prepend the purifier $0$ to $\pI$ if $|\pI|$ is odd.  Specifically, the $\G$ ordering would then be
\begin{equation}\label{eq:Gorder}
\{ 01,\, 02,\, 12,\, 03,\, 13,\, 23,\, 0123,\, 04,\, 14,\, 24,\, 0124,\, 34,\, 0134,\, 0234,\, 1234, \ldots \} \ .
\end{equation}	
Writing the conversion matrices with the rows and columns ordered according to \eqref{eq:Iorder} and \eqref{eq:Gorder}, it is easy to verify that the ${\alpha}$ matrix is upper-triangular due to \eqref{eq:Iofgmincard}, so   the corresponding ${\beta} = {\alpha}^{-1} $ matrix must be likewise upper-triangular.  Indeed, we have the slightly stronger condition
\begin{equation} \label{eq:gofImincard}
\beta_{\G}{}^{\pI} =0 \quad  \text{whenever}\quad  \gI\setminus 0 \nsubseteq \pI \ ,  
\end{equation}	
which requires some entries above the diagonal to vanish as well, analogous to the case for $\alpha$, i.e.\ condition \eqref{eq:Iofgzero}. Condition \eqref{eq:gofImincard} follows from the additional requirement that the system is permutation symmetric, so under permutations (i.e.\ simple relabeling of the rows and columns) the $\alpha$ and $\beta$ matrices have to remain upper triangular.

Using \eqref{eq:gofImincard}, we see that if $|\pI| < |\gI \setminus 0|$, then $\gI \setminus 0 \nsubseteq \pI$, which means $\b_\G{}^\pI = 0$. Substituting \eqref{eq:KfromI} into \eqref{eq:KineqR}, we then obtain
\begin{align}
\begin{split}
	\Q = \sum_{\gI \subseteq [\N+1], |\G\setminus 0| > \R }\lambda_{\gI}\, \k_{\gI}^{(\N)} 
	&= \sum_{\pI \subseteq [\N], |\pI| \geq |\gI \setminus 0|} \left( \sum_{\gI \subseteq [\N+1], |\G\setminus 0|> \R }\lambda_{\gI} \, \beta_{\G}{}^{\pI}\right) \I_\pI\ .
\end{split}
\end{align}
This shows that $\Q$ only involves $\I_\pI$ with $|\pI| \geq \R+1$ and is therefore $\R$-balanced. Hence, the definition of $\R$-balance in the K-basis is given as follows:
\begin{defn}\label{def:RbalK}
An $\N$-party information quantity $\Q$ is $\R$-balanced iff in its K-basis representation (last expression in \eqref{eq:SIKineq}) it does not contain terms $ \k_{\gI}^{(\N)} $ with $|\G\setminus 0|\le \R$.
\end{defn}	
\noindent
Note that the purifier $0$ appears explicitly in this definition because although the K-basis is purification-symmetric itself, the definition of $\R$-balance, formulated in the I-basis, is not purification symmetric and therefore underlies the motivation for desiring superbalance rather than just balance.

Specializing now to superbalance, $\R=2$, we obtain the following corollary, which we will use in the following subsection to prove that all multipartite (i.e.\ other than SA) HEIs are superbalanced.
\begin{coro}\label{def:superbalK}
An $\N$-party information quantity $\Q$ (or the corresponding HEI, $\Q\ge0$) is superbalanced iff its K-basis representation does not contain $\k_{\ii\jj}^{(\N)}$ for any $\ii,\jj\in[\N+1]$.
\end{coro}	
%

\subsection{Proving superbalance}
\label{ss:proof}

Before proving the main result of this paper, let us warm up making a simple observation that will be useful later.  The K-basis representation of SA is\footnote{
We now revert back to the notation in \cite{He:2019ttu}, where the purifier is labeled by $\N+1$ instead of $0$, so $\ii$ and $\jj$ in \eqref{eq:SAK} run from $1$ to $\N+1$.
}  
\begin{equation}\label{eq:SAK}
\k_{\ii\jj}^{(\N)} \geq 0 \ .
\end{equation}	
This follows from the I-basis representation of SA \eqref{sa}, namely $\I_{ij}\ge 0$, and the observation that in the expansion \eqref{eq:IfromK} for $\pI=ij$,
\begin{align}\label{I2}
	\alpha_{ij}{}^{\G}
	=\I_{ij}\big|_{\KB{\G}} = \begin{cases}
	2, & \G = \{ij\} \\
	0, & \G \not= \{ij\}\ ,
\end{cases}
\end{align}
which implies that $\I_{ij} = 2\k_{ij}^{(\N)}$.  The full orbit under purifications corresponds to enlarging $i \to \ii$ and $j \to \jj$, which gives \eqref{eq:SAK}.

The claim that non-redundant HEIs except SA are superbalanced is a direct consequence of the following theorem, which is interesting in its own right, as it constrains the possible form extreme rays may have in the $\N$-party entropy cone. Begin by considering an $\N$-party system, and denote the HEC by $\CC^\N$. Let $\KB{\G}$ label the K-basis vectors (the PTs), and let $\hat v^{(\N)}$ be any other (i.e.\ non-PT) extreme ray of $\CC^\N$. Note that PT extreme rays that are not BPs have vanishing mutual information between the monochromatic subsystems, and hence cannot involve BPs.

\begin{theorem} \label{thm:main}
The non-Bell pair (BP) extreme rays do not involve BPs.  In particular, any non-BP extreme ray $\hat v^{(\N)}$ of $\CC^\N$ written in the K-basis
\begin{align}
\label{eq:vK}
	\hat v^{(\N)} = \sum_{\G \subseteq [\N+1]} \av_{\G}^{(\N)} \, \KB{\G} \quad\text{where}\quad \av_\G^{(\N)} \in \mathbb R\ ,
\end{align}
must satisfy 
\begin{align}\label{claim}
	\av_{\G}^{(\N)}=0 \quad \forall\  |\G| = 2 \ .
\end{align}
\end{theorem}
\begin{proof} 

First, observe that since all the PTs are basis vectors in entropy space, if $\hat v^{(\N)} = \hat g^{\G'}$ for some PT $\G'$, then the linear independence of the basis implies
\begin{align}
	\av_\G^{(\N)} = \begin{cases}
	1 & \G = \G' \\
	0 & \G \not= \G'\ .
\end{cases}
\end{align}	
This in particular means that none of the higher-party PTs involve BPs, so we will henceforth assume $\hat v^{(\N)}$ is not a PT. 

We now proceed via induction. The claim is trivially true for $\N=3$, as the extreme rays of $\CC^{3}$ consist of six BPs and a single PT$_4$, so the only non-BP $\hat v^{(3)}$ does not involve a BP.
Our induction hypothesis is thus that the claim is true for $\CC^{\N}$ with some fixed $\N \geq 3$. To show that the claim continues to hold for $\CC^{\N+1}$, we employ a proof by contradiction.

Suppose there exists some non-PT extreme ray $\hat v^{(\N+1)}$ of $\CC^{\N+1}$ such that
\begin{align}
	\hat v^{(\N+1)} = \sum_{\G \subseteq [\N+1]} \av_{\G}^{(\N+1)} \, \KB{\G}\ ,
\end{align}
and $\av_{\ii\jj}^{(\N+1)} \not= 0$ for some $\ii,\jj \in [\N+1]$. We may without loss of generality relabel the parties to get $\av_{12}^{(\N+1)} \not= 0$, so that 
\begin{align}
	\hat v^{(\N+1)} = \av_{12}^{(\N+1)} \, \KB{\{12\}} + \sum_{\G \subseteq[ \N+1], \G \not= \{12\}}\av_\G^{(\N+1)} \, \KB{\G}\ , \qquad\text{where}\quad \av_{12}^{(\N+1)} \not= 0\ .
\end{align}
Evaluating $\I_{12}$ of $\hat v^{(\N+1)}$ using \eqref{I2}, we obtain
\begin{align}\label{12coeff}
	\left.\I_{12}\right|_{\hat v^{(\N+1)}} = 2\, \av_{12}^{(\N+1)}\ .
\end{align}
Since 
$\hat v^{(\N+1)}$  is in $\CC^{\N+1}$, it must obey SA, which means that $\av_{12}^{(\N+1)} > 0$. 

Any entropy vector in $\CC^{\N+1}$ can be associated to some static asymptotically AdS geometry by the arguments  
of \cite{Bao:2015bfa}, which can in turn be represented by a discrete graph. The graph for a static geometry at a constant time slice is constructed by partitioning the spatial slice using minimal (RT) surfaces and associating a vertex to each resulting codimension-0 region. Regions adjoining the boundary yield the boundary vertices associated to the parties (including the purifier), while the remaining regions are represented by internal vertices. The links between the vertices correspond to pieces of the RT surfaces, and are endowed with a weight given by the corresponding surface area.%
\footnote{
It was shown in \cite{Bao:2015bfa,Cui:2018dyq} 
that such graphs must obey MMI. This is a necessary assumption, since in the general quantum entropy cone, not all inequalities besides SA are superbalanced. One notable example is strong subadditivity, whose purification (known as weak monotinicity) is not balanced.} Below, in \cref{f:graphrep} we have schematically drawn the graph representing $\hat v^{(\N+1)}$, with the edges given by weights $n_{\ii}$.\footnote{
Fig.~6 of \cite{Bao:2015bfa} describes several graph transformations that preserve all the discrete entropies, and we may require such transformations to obtain the graph in \cref{f:graphrep}. 
 }
The region $G$ in the graph is simply some unspecified set of legs and vertices, and the weights of the legs emanating from $\CA_\ii$ are denoted by $n_\ii$, which we can take to be $\S_\ii$ via the entropy-preserving graph transformations of \cite{Bao:2015bfa}.

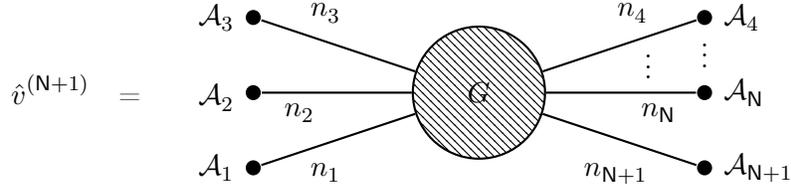
\begin{figure}[h]
\centering
  \begin{tikzpicture}
  [
       thick,
       acteur/.style={
         circle,
         fill=black,
         thick,
         inner sep=2pt,
         minimum size=0.2cm
       }
     ] 
    \begin{feynman}
      \vertex[blob,label={center:$G$},minimum size=50] (m) at ( 0, 0) {\contour{white}{}};
      \vertex (a1) at (-3,-1) [acteur,label=left:$\CA_{1}$]{};
      \vertex (a2) at (-3, 0) [acteur,label=left:$\hat v^{(\N+1)} \quad\text{=} \quad\quad  \CA_{2}$]{};
	\vertex (a3) at (-3,1) [acteur,label=left:$\CA_{3}$]{};
      \vertex (a4) at (3,1) [acteur,label=right:$\CA_{4}$]{};
      \vertex (dot) at (3, 0.6) {$\vdots$};
      \vertex (dot2) at (2.25, 0.45) {$\vdots$};
      \vertex (a5) at (3,0) [acteur,label=right:$\CA_{\N}$]{};
      \vertex (a6) at (3,-1) [acteur,label=right:$\CA_{\N+1}$]{} ;
      \diagram* {
        (a1) -- [edge label'=$n_{1}$,pos=0.25] (m),
        (a2) -- [edge label'=$n_2$,pos=0.25] (m),
	 (a3) -- [edge label=$n_3$,pos=0.25] (m),
	 (a4) -- [edge label'=$n_4$,pos=0.25] (m),
        (a5) -- [edge label=$n_{\N}$,pos=0.25] (m),
	 (a6) -- [edge label=$n_{\N+1}$,pos=0.25] (m) 
             };
    \end{feynman}
  \end{tikzpicture}
  \caption{The graph representation of the entropy vector $\hat v^{(\N+1)}$.}
  \label{f:graphrep}
\end{figure}

We now coarse-grain parties $\{\N,\N+1\} \to \{\N\}$ by simply joining legs connected to vertices $\N$ and $\N+1$. We can view the coarse-graining as a linear map $L$ acting on $\hat v^{(\N+1)}$, and the result is\footnote{
We can alternatively think of this as just relabeling $\CA_{\N+1}$ to be $\CA_\N$, so the new graph still corresponds to a holographic entropy vector.} 

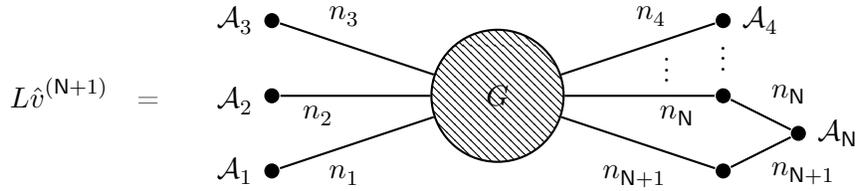
\begin{figure}[H]
\centering
  \begin{tikzpicture}
  [
       thick,
       acteur/.style={
         circle,
         fill=black,
         thick,
         inner sep=2pt,
         minimum size=0.2cm
       }
     ] 
    \begin{feynman}
      \vertex[blob,label={center:$G$},minimum size=50] (m) at ( 0, 0) {\contour{white}{}};
      \vertex (a1) at (-3,-1) [acteur,label=left:$\CA_{1}$]{};
      \vertex (a2) at (-3, 0) [acteur,label=left:$L\hat v^{(\N+1)} \quad\text{=} \quad\quad \CA_{2}$]{};
	\vertex (a3) at (-3,1) [acteur,label=left:$\CA_{3}$]{};
      \vertex (a4) at (3,1) [acteur,label=right:$\CA_{4}$]{};
      \vertex (dot) at (3, 0.6) {$\vdots$};
      \vertex (dot2) at (2.25, 0.45) {$\vdots$};
      \vertex (a5) at (3,0) [acteur]{};
      \vertex (a6) at (3,-1) [acteur]{} ;
      \vertex (b) at (4,-0.5) [acteur,label=right:$\CA_\N$]{};
      \diagram* {
        (a1) -- [edge label'=$n_{1}$,pos=0.25] (m),
        (a2) -- [edge label'=$n_2$,pos=0.25] (m),
	 (a3) -- [edge label=$n_3$,pos=0.25] (m),
	 (a4) -- [edge label'=$n_4$,pos=0.25] (m),
       (a5) -- [edge label=$n_{\N}$,pos=0.25] (m),
	 (a6) -- [edge label=$n_{\N+1}$,pos=0.25] (m),
	 (a5) --  [edge label=$n_{\N}$] (b),
	 (a6) -- [edge label'=$n_{\N+1}$] (b)
             };
    \end{feynman}
  \end{tikzpicture}
  \caption{We coarse-grain $\hat v^{(\N+1)}$ to obtain an entropy vector in $\CC^{\N}$ from one in $\CC^{\N+1}$.}
\end{figure}

\noindent This is now a graph representation of an entropy vector in $\CC^\N$. Explicitly, we have for some coefficients $\av_\G^{(\N)}$
\begin{align}\label{coarse}
	L\hat v^{(\N+1)} = \av_{12}^{(\N)} \, \KB{\{12\}} + \sum_{\G \subseteq [\N],\G \not= \{12\}}\av_\G^{(\N)} \, \KB{\G}\ ,
\end{align}
where the sum is now over all even-party PTs of the $\CC^\N$ cone, rather than the $\CC^{\N+1}$ cone. From the graphs, it is obvious that $\left.\I_{12}\right|_{\hat v^{(\N+1)}} = \left.\I_{12}\right|_{L\hat v^{(\N+1)}}$ since $\S_1$, $\S_2$, and $\S_{12}$ are the same before and after coarse-graining. Thus,
\begin{align}
	\av_{12}^{(\N)} = \frac{1}{2} \, \left.\I_{12}\right|_{L\hat v^{(\N+1)}} = \frac{1}{2} \, \left.\I_{12}\right|_{\hat v^{\N+1}} = \av_{12}^{(\N+1)}   > 0\ .
\end{align}
However, since $L\hat v^{(\N+1)} \in \CC^\N$, by the induction hypothesis this means  that none of the non-BP extreme rays in $\CC^N$ involve $\KB{\{12\}}$, and so the $\KB{\{12\}}$ term on the right-hand-side of \eqref{coarse} can be isolated as the $\{12\}$ Bell pair and correspondingly be rendered in the graph as simply an edge connecting $\CA_1$ to $\CA_2$ (as was shown in \cref{f:Kbasis}).
Because none of the edges outside $G$ join $\CA_1$ to $\CA_2$, $\KB{\{12\}}$ must lie within $G$. Thus, $G$ must be realizable as some graph $G'$ plus an edge joining $\CA_1$ and $\CA_2$ with weight $\av_{12}^{(\N+1)}$, i.e.

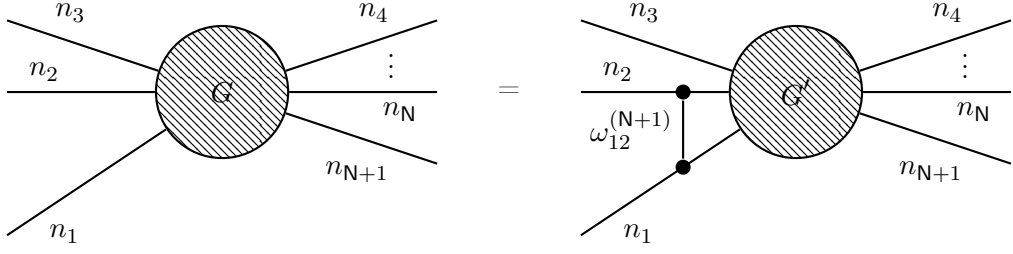
\begin{figure}[H]
\begin{center}
\begin{tikzpicture}
  [
       thick,
       acteur/.style={
         circle,
         fill=black,
         thick,
         inner sep=2pt,
         minimum size=0.2cm
       }
     ] 
    \begin{feynman}
      \vertex[blob,label={center:$G$},minimum size=50] (m) at ( 0, 0) {\contour{white}{}};
      \vertex (a1) at (-3,-2) []{};
      \vertex (a2) at (-3, 0) []{};
	\vertex (a3) at (-3,1) []{};
      \vertex (a4) at (3,1) []{};
      \vertex (dot2) at (2.25, 0.45) {$\vdots$};
      \vertex (a5) at (3,0) [label=right: \quad \text{=}]{};
      \vertex (a6) at (3,-1) []{} ;
      \diagram* {
        (a1) -- [edge label'=$n_{1}$,pos=0.2] (m),
        (a2) -- [edge label=$n_2$,pos=0.25] (m),
	 (a3) -- [edge label=$n_3$,pos=0.25] (m),
	 (a4) -- [edge label'=$n_4$,pos=0.25] (m),
       (a5) -- [edge label=$n_{\N}$,pos=0.25] (m),
	 (a6) -- [edge label=$n_{\N+1}$,pos=0.25] (m),
             };
    \end{feynman}
  \end{tikzpicture} 
  \begin{tikzpicture}
  [
       thick,
       acteur/.style={
         circle,
         fill=black,
         thick,
         inner sep=2pt,
         minimum size=0.2cm
       }
     ] 
    \begin{feynman}
      \vertex[blob,label={center:$G'$},minimum size=50] (m) at ( 0, 0) {\contour{white}{}};
      \vertex (a1) at (-3,-2) []{};
      \vertex (a2) at (-3, 0) [label=left:\quad]{};
	\vertex (a3) at (-3,1) []{};
      \vertex (a4) at (3,1) []{};
      \vertex (dot2) at (2.25, 0.45) {$\vdots$};
      \vertex (a5) at (3,0) []{};
      \vertex (a6) at (3,-1) []{} ;
      \vertex (c1) at (-1.5, -1) [acteur]{};
      \vertex (c2) at (-1.5,0) [acteur]{};
      \diagram* {
        (a1) -- [edge label'=$n_{1}$,pos=0.2] (m),
        (a2) -- [edge label=$n_2$,pos=0.25] (m),
	 (a3) -- [edge label=$n_3$,pos=0.25] (m),
	 (a4) -- [edge label'=$n_4$,pos=0.25] (m),
       (a5) -- [edge label=$n_{\N}$,pos=0.25] (m),
	 (a6) -- [edge label=$n_{\N+1}$,pos=0.25] (m),
	 (c1) -- [edge label= $\av_{12}^{(\N+1)}$] (c2)
             };
    \end{feynman}
  \end{tikzpicture}
  \end{center}
  \caption{Explicitly showing the BP in $G$.}
\end{figure}

\noindent Substituting this into our graph representation of $\hat v^{(\N+1)}$, we get

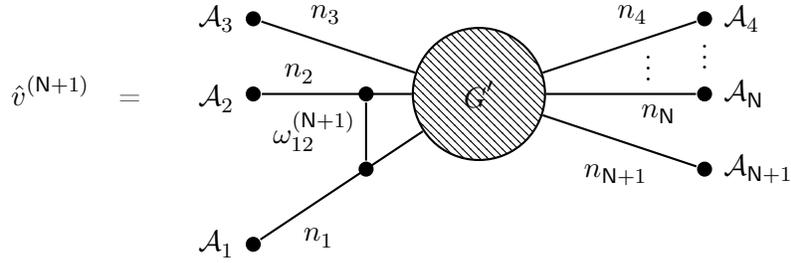
\begin{figure}[H]
\centering
  \begin{tikzpicture}
  [
       thick,
       acteur/.style={
         circle,
         fill=black,
         thick,
         inner sep=2pt,
         minimum size=0.2cm
       }
     ] 
    \begin{feynman}
      \vertex[blob,label={center:$G'$},minimum size=50] (m) at ( 0, 0) {\contour{white}{}};
      \vertex (a1) at (-3,-2) [acteur,label=left:$\CA_{1}$]{};
      \vertex (a2) at (-3, 0) [acteur,label=left:$\hat v^{(\N+1)} \quad\text{=} \quad\quad\CA_{2}$]{};
	\vertex (a3) at (-3,1) [acteur,label=left:$\CA_{3}$]{};
      \vertex (a4) at (3,1) [acteur,label=right:$\CA_{4}$]{};
      \vertex (dot) at (3, 0.6) {$\vdots$};
      \vertex (dot2) at (2.25, 0.45) {$\vdots$};
      \vertex (a5) at (3,0) [acteur,label=right:$\CA_\N$]{};
      \vertex (a6) at (3,-1) [acteur,label=right:$\CA_{\N+1}$]{} ;
      \vertex (c1) at (-1.5, -1) [acteur]{};
      \vertex (c2) at (-1.5,0) [acteur]{};
      \diagram* {
        (a1) -- [edge label'=$n_{1}$,pos=0.2] (m),
        (a2) -- [edge label=$n_2$,pos=0.25] (m),
	 (a3) -- [edge label=$n_3$,pos=0.25] (m),
	 (a4) -- [edge label'=$n_4$,pos=0.25] (m),
       (a5) -- [edge label=$n_{\N}$,pos=0.25] (m),
	 (a6) -- [edge label=$n_{\N+1}$,pos=0.25] (m),
	 (c1) -- [edge label= $\av_{12}^{(\N+1)}$] (c2)
             };
    \end{feynman}
  \end{tikzpicture}
  \caption{The graph of $\hat v^{(\N+1)}$ with edge representing the BP $\hat g^{\{12\}}$ explicitly shown.}
\end{figure}

\noindent As $\av_{12}^{(\N+1)} > 0$, this means we can also consider the entropy vector $\vec w$ given by the graph\footnote{
We are implicitly assuming here that $\vec{w}$ lies within the entropy cone since it has a graphical representation. This can be a subtle issue -- see  \cite{Marolf:2017shp} for a discussion.}

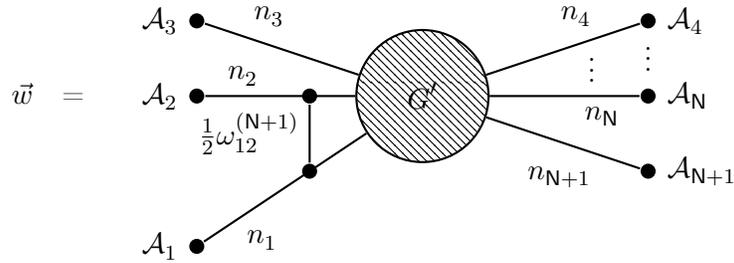
\begin{figure}[h]
\centering
  \begin{tikzpicture}
  [
       thick,
       acteur/.style={
         circle,
         fill=black,
         thick,
         inner sep=2pt,
         minimum size=0.2cm
       }
     ] 
    \begin{feynman}
      \vertex[blob,label={center:$G'$},minimum size=50] (m) at ( 0, 0) {\contour{white}{}};
      \vertex (a1) at (-3,-2) [acteur,label=left:$\CA_{1}$]{};
      \vertex (a2) at (-3, 0) [acteur,label=left:$\vec w \quad\text{=} \quad\quad \CA_{2}$]{};
	\vertex (a3) at (-3,1) [acteur,label=left:$\CA_{3}$]{};
      \vertex (a4) at (3,1) [acteur,label=right:$\CA_{4}$]{};
      \vertex (dot) at (3, 0.6) {$\vdots$};
      \vertex (dot2) at (2.25, 0.45) {$\vdots$};
      \vertex (a5) at (3,0) [acteur,label=right:$\CA_\N$]{};
      \vertex (a6) at (3,-1) [acteur,label=right:$\CA_{\N+1}$]{} ;
      \vertex (c1) at (-1.5, -1) [acteur]{};
      \vertex (c2) at (-1.5,0) [acteur]{};
      \diagram* {
        (a1) -- [edge label'=$n_{1}$,pos=0.2] (m),
        (a2) -- [edge label=$n_2$,pos=0.25] (m),
	 (a3) -- [edge label=$n_3$,pos=0.25] (m),
	 (a4) -- [edge label'=$n_4$,pos=0.25] (m),
       (a5) -- [edge label=$n_{\N}$,pos=0.25] (m),
	 (a6) -- [edge label=$n_{\N+1}$,pos=0.25] (m),
	 (c1) -- [edge label= $\frac{1}{2}\av_{12}^{(\N+1)}$] (c2)
             };
    \end{feynman}
  \end{tikzpicture}
  \caption{The graph of $\vec w$.}
\end{figure}

\noindent Because $\KB{\{12\}}$ is simply the graph connecting vertices $\CA_1$ and $\CA_2$ with weight $1$, it follows that
\begin{align}
	\vec w  + \frac{1}{2}\av_{12}^{(\N+1)}\KB{12} = \hat v^{(\N+1)}\ .
\end{align}
Since $\av_{12}^{(\N+1)} > 0$ and both $\KB{12}$ and $\vec w$ are in $\CC^{\N+1}$, this implies $\hat v^{(\N+1)}$ is a positive sum of two distinct vectors in the cone, contradicting the fact that it is an extreme ray. This completes the proof by induction.
\end{proof}

Having established Theorem~\ref{thm:main}, the superbalanced nature of the HEIs is a straightforward corollary.

\begin{coro}\label{maincor}
	All non-redundant HEIs that are not SA are superbalanced.
\end{coro}
\begin{proof}
	Consider the $\N$-party entropy space. By Corollary~\ref{def:superbalK}, it suffices to prove that all of the non-redundant HEIs except SA do not involve $\k_{\ii\jj}^{(\N)}$ for any $\ii,\jj \in [\N+1]$ when written in the K-basis. We will again proceed by contradiction.
	
Suppose that some non-redundant HEI involves $\k_{\ii\jj}^{(\N)}$. We can always relabel the parties so that the inequality involves $\k_{12}^{(\N)}$, i.e.
\begin{align}\label{primhei}
	\sum_{\G \not= \{12\}} \l_\G  \, \k_{\G}^{(\N)} + \l_{12} \, \k_{12}^{(\N)} \geq 0, \quad\text{where $\l_{12} \not= 0$\ .}
\end{align}
This inequality holds for all $\vec\S \in \CC^\N$. However, because it is non-redundant, there exists some $\vec\S \in \CC^\N$ with coefficients $\k_{\G}^{(\N)}$, which when expressed in the K-basis satisfies \eqref{primhei}, but
\begin{align}\label{primhei2}
	\sum_{\G \not= \{12\}}\l_{\G} \, \k^{(\N)}_{\G} < 0\ .
\end{align}
From \eqref{primhei} and \eqref{primhei2}, it follows that
\begin{align}\label{notincone}
\begin{split}
	\vec v_1 &\equiv \sum_{\G\not=\{12\}}  \, \k_{\G }^{(\N)}\KB{\G } \notin \CC^{\N} \\
	\vec v_2 &\equiv \sum_{\G \not=\{12\}}  \, \k_{\G }^{(\N)}\KB{\G } + \k_{12}^{(\N)}\KB{12} \in \CC^\N\ .
\end{split}
\end{align}
Because $\vec v_2 \in \CC^\N$, it must be a positive sum of extreme rays. However, by Theorem~\ref{thm:main}, none of the non-BP extreme rays involve $\KB{12}$, so this means $\vec v_1$ must also be a positive sum of extreme rays and hence lie $\CC^\N$. This is a contradiction, completing the proof.
\end{proof}

\section{Discussion} 
\label{s:discussion}

We have proved that excluding SA, all non-redundant HEIs are superbalanced. We demonstrated this by first proving that when we write each of the extreme rays of the HEC as a linear combination of the PTs, none of the non-BP extreme rays involve any BPs (see Theorem~\ref{thm:main}). Crucial to the proof of Theorem~\ref{thm:main} is the fact  that we can coarse-grain parties $\N$ and $\N+1$ without altering the internal structure of $G$. This is specific only to holographic entropy vectors expressible via a graph representation. In this sense, there is a notion of bulk locality regarding entanglement entropy, where we can change the RT surfaces near the boundary without affecting the RT surfaces elsewhere. It would be interesting to explore whether or not this ``bulk locality'' can be used to further understand the nature of the HEIs and thereby further constrain the HEC.

We would like to conclude by commenting that it is possible to reverse the arguments in the previous section to show the following. 

\begin{rmk}
	All non-redundant HEIs except for SA are superbalanced iff the non-BP extreme rays of the HEC $\CC^\N$, when expressed in the K-basis, do not involve any BPs.
\end{rmk}
\begin{proof}
	The reverse direction is proven by Corollary~\ref{maincor}, so it suffices to prove the forward direction. Suppose that all non-redundant HEIs excepting SA are superbalanced, and that (relabeling the parties if necessary) there exists an extreme ray $\vec v$ of $\CC^\N$ such that
	\begin{align}
	\vec v = \av_{12} \, \KB{\{12\}} + \sum_{\G \not= \{12\}} \av_\G \, \KB{\G}\ , \quad\text{where $\av_\G \in \mathbb R$ and $\av_{12} \not= 0$\ .}
\end{align}
Note that $\I_{12}|_{\vec v} = \av_{12}$, so $\av_{12} \geq 0$ by SA. Furthermore, as $\vec v$ is an extreme ray, it cannot be written as a sum of two other vectors in the HEI $\CC^\N$. In particular, $\av_{12}\, \KB{12} \in \CC^\N$, so
\begin{align}
	\vec w \equiv \sum_{\G \not= \{12\}} \av_\G \, \KB{\G} \notin \CC^\N\ .
\end{align}
On the other hand, since $\vec w = \vec v - \av_{12}\, \KB{12}$, $\vec w$ obeys any HEI not involving $\k_{12}^{(\N)}$ since $\vec v$ obeys it. By \eqref{eq:SAK} and Corollary~\ref{def:superbalK}, the only HEI involving $\k_{12}^{(\N)}$ is the SA $\k_{12}^{(\N)} \geq 0$, and obviously $\vec w$ obeys this SA as well. Thus, $\vec w$ obeys all the non-redundant HEIs and hence must lie in $\CC^\N$, a contradiction.
\end{proof}

This demonstrates that the superbalanced nature of the non-redundant multipartite HEIs is equivalent to the fact that the only extreme rays of the HEC with non-vanishing monochromatic mutual information are the BPs. It is not a-priori obvious why this had to be the case, and leaves open the question of whether extreme rays of the HEC have other simple properties as well. For instance, a quick examination of Table~\ref{t:N=5} shows that none of the other HEIs for $\N=5$ are 3-balanced, and it is natural to wonder this is a general feature among all HEIs, or whether there are certain classes of HEIs that are $\R$-balanced for some $\R \geq 3$.  Moreover, our arguments using graph models were valid in the static context where the RT prescription applies, whereas the work of  \cite{Hubeny:2018trv,Hubeny:2018ijt} suggests a deeper and more general structure, dubbed the holographic entropy arrangement and consists of primitive information quantities (the sign-definite subset of which forms the holographic entropy polyhedron).  Since all known primitive quantities excepting SA are superbalanced (even the non-sign-definite ones), it would be useful to prove this result directly in the proto-entropy language.
We will leave this, as well as a deeper understanding of why HEIs are superbalanced, for future investigations.

\section*{Acknowledgements}
It is a pleasure to thank S.~Hern\'andez Cuenca and M.~Rota for illuminating discussions on the holographic entropy polyhedron and collaboration on related topics, M.~Headrick and M.~Walter for stimulating earlier discussions, and M.~Rota for useful comments on the draft.
We would like to thank KITP, UCSB for hospitality during the workshop ``Gravitational Holography'', where the research was supported in part by the National Science Foundation under Grant No.~NSF PHY1748958 to the KITP. 
The authors are supported  by U.S.\ Department of Energy grant DE-SC0009999 and by funds from the University of California. 

\appendix

\section{Conversion between I and K bases}
\label{a:IKconf}

In \cref{ss:Iofg} we presented an argument relating the entropy vectors in the I  and K bases. We now give a more concrete translation between the two bases. This provides an alternative (but essentially equivalent) proof of Corollary~\ref{def:superbalK}, which states that an HEI is superbalanced iff it does not contain $\k_{\ii\jj}^{(\N)}$ in the K-basis. 

Consider \eqref{eq:IfromK}, which we reproduce here for convenience:
\begin{align}
	\I_\pI = \sum_{\G \subseteq [\N+1]} \alpha_{\pI}{}^{\gI } \, \k_\G\ .
\end{align}
We want to determine explicitly the entries of the matrix $\a$. Recalling \eqref{eq:alphadef}, which states that $\a_\pI{}^{\gI} = \left.\I_\pI\right|_{\hat g^\G}$, and letting $|\pI \cap \G| = m$, we obtain via combinatorics
\begin{align}\label{ItoK}
\begin{split}
	\left.\I_\pI\right|_{\KB{\G}} = \sum_{\pI' \subseteq \pI} (-1)^{|\pI'|+1} \, \left.\S_{\pI'}\right|_{\KB{\G}} &= \sum_{k=1}^{|\pI|}\sum_{r=1}^k (-1)^{k+1}  \, \min(r,|\G|-r)\binom{m}{r}\binom{|\pI|-m}{k-r} \\
	&= \sum_{r=1}^{|\pI|} (-1)^r \, \min(r,|\G|-r)\binom{m}{r} \sum_{k=0}^{|\pI|-r} (-1)^{k+1} \binom{|\pI|-m}{k}\ .
\end{split}
\end{align}
In particular, note that for $|\pI| > m$, we have
\begin{align}
\begin{split}
	\left.\I_\pI\right|_{\KB{\G}} = \sum_{r=1}^{|\pI|} (-1)^r \, \min(r,|\G|-r)\binom{m}{r} \sum_{k=0}^{|\pI|-m} (-1)^{k+1} \binom{|\pI|-m}{k} = 0\ ,
\end{split}
\end{align}
where we used the fact the second sum vanishes when $|\pI| - m > 0$, whereas for
$|\pI| = m$, we have
\begin{align}
\begin{split}
	\left.\I_\pI\right|_{\KB{\G}} &= \sum_{r=1}^{|\pI|} (-1)^{r+1}  \, \min(r,|\G|-r)\binom{|\pI|}{r}\ .
\end{split}
\end{align}
To summarize, it follows
\begin{align}\label{zeros}
	\a_{\pI}{}^\gI = \left.\I_\pI\right|_{\KB{\G}} = \begin{cases}
	0 & |\pI| > m \\
	\sum_{r=1}^{|\pI|} (-1)^{r+1}  \, \min(r,|\G|-r)\binom{|\pI|}{r} & |\pI| = m\ .
\end{cases}
\end{align}
This is precisely equivalent to \eqref{eq:Iofgmincard}, and can be used to prove the forward direction of Corollary~\ref{def:superbalK}.

We can also translate an HEI in K-basis to I-basis. In particular, given
\begin{align}\label{Iineq2}
	\sum_{\gI \subseteq [\N+1]} \lambda_\G \, \k^{(\N)}_\G = \sum_{\pI \subseteq [\N]} \sigma_\pI\I_\pI \geq 0 \ ,
\end{align}
we can determine $\sigma_\pI$ in terms of $\lambda_\G$ using the following procedure.
Consider now the entropy vector $\KB{\{i(\N+1)\}}$, where $i \in [\N]$. Using \eqref{zeros}, we observe that for any $\pI \subseteq [\N]$,
\begin{align}
	\left.\I_\pI\right|_{\KB{\{i(\N+1)\}}} = \begin{cases}
	0 & \pI \not= i \\
	1 & \pI = i\ .
\end{cases}
\end{align}
Thus, evaluating both sides of \eqref{Iineq2} on $\KB{\{i(\N+1)\}}$, we obtain for all $i \in [\N]$,
\begin{align}\label{con2}
	\sigma_i = \lambda_{i(\N+1)}\ .
\end{align}

Next, consider the entropy vector $\KB{\{ij\}}$, where $i,j \in [\N]$. Again using \eqref{zeros}, we observe that for any $\pI \subseteq [\N]$ and $i,j \in [\N]$,
\begin{align}
	\left.\I_\pI\right|_{\KB{\{ij\}}} = \begin{cases}
	1 & \pI = \text{$i$ or $j$} \\
	2 & \pI = \{i,j\} \\
	0 & \text{otherwise}.
\end{cases}
\end{align}
Evaluating both sides of \eqref{Iineq2} on this $\KB{\{ij\}}$, we obtain
\begin{align}\label{con1}
	\sigma_i + \sigma_j + 2\sigma_{ij} = \lambda_{ij}\ .
\end{align}
By \eqref{con2}, this implies $\sigma_{ij} = \frac{1}{2}\left(\lambda_{ij} - \lambda_{i(\N+1)} - \lambda_{j(\N+1)}\right)$ for all $i,j \in [\N]$. Indeed, note that if $\lambda_{\ii\jj} = 0$ for all $\ii,\jj \in [\N+1]$, this immediately implies $\sigma_i = \sigma_{ij} = 0$ for all $i,j \in [\N]$, i.e. the HEI is superbalanced. This proves the reverse direction of Corollary~\ref{def:superbalK}.

We can continue the above procedure to write the $\sigma_\pI$'s in terms of $\lambda_\gI$'s by considering the entropy vectors $\KB{\{i_1\cdots i_{2k-1}(\N+1)\}}$ and $\KB{\{i_1\cdots i_{2k}\}}$ for all $k \leq n$, where $i_1,\ldots,i_{2k} \in [\N]$. By repeatedly evaluating \eqref{Iineq2} on these PTs and using \eqref{zeros}, we are able to write linear equations relating $\sigma_\pI$ and $\lambda_\gI$, in the same spirit as \eqref{con2} and \eqref{con1}.  


\providecommand{\href}[2]{#2}\begingroup\raggedright\endgroup

\end{document}